\documentclass[copyright,creativecommons]{eptcs}
\providecommand{\event}{GandALF 2019}

\usepackage{etex}

\usepackage[utf8]{inputenc}
\usepackage[T1]{fontenc}
\usepackage{anyfontsize}

\usepackage{microtype}
\usepackage[british]{babel}
\usepackage{soul}

\usepackage[framemethod=tikz]{mdframed}
\usepackage{wrapfig}

\usepackage{amssymb}
\usepackage{amsmath}
\usepackage{amsfonts}
\usepackage{stmaryrd}

\usepackage{calrsfs}
\DeclareMathAlphabet{\pazocal}{OMS}{zplm}{m}{n}

\usepackage{amsthm}
\newtheorem{lemma}{Lemma}

\usepackage{tabu,booktabs}
\usepackage{multicol,multirow}
\usepackage[inline,shortlabels]{enumitem}

\usepackage{tikz}
\usetikzlibrary{calc, patterns, shapes, shapes.misc, intersections, positioning, arrows.meta}
\usetikzlibrary{decorations.pathreplacing, backgrounds}
\usetikzlibrary{hobby}

\usepackage[noend]{algorithm2e}
\AlgoDontDisplayBlockMarkers
\DontPrintSemicolon
\LinesNumbered
\SetAlgoCaptionLayout{normal}
\SetAlgoCaptionSeparator{}
\SetAlCapSkip{.5em}
\makeatletter
\renewcommand{\algocf@makecaption}[2]{%
	\addtolength{\hsize}{1.5\algomargin}%
	\parbox[t]{\hsize}{\algocf@captiontext{#1:}{#2}}%
	\addtolength{\hsize}{-1.5\algomargin}%
}
\makeatother
\SetArgSty{textup}

\SetCommentSty{mycommfont}
\SetKwComment{tcp}{}{}
\SetKwProg{Def}{def}{:}{}
\SetKwProg{Thread}{thread}{:}{}
\SetKwIF{If}{ElseIf}{Else}{if}{:}{elif}{else:}{}
\SetKwFor{For}{for}{:}{}
\SetKwFor{ForAll}{forall}{:}{}
\SetKwFor{While}{while}{:}{}
\SetKwIF{Catch}{}{Try}{catch}{:}{}{try:}{}
\SetKwFor{Loop}{loop:}{}{}
\SetKwFor{DoPar}{do in parallel:}{}{}
\SetKw{Break}{break}
\SetKw{Raise}{raise}
\SetKw{False}{false}
\SetKw{True}{true}
\SetKw{Continue}{continue}
\SetKw{KwTo}{in}
\SetKwFunction{DFI}{dfi}
\SetKwFunction{Winner}{winner}
\SetKwFunction{Onestep}{onestep}

\newcommand{\cupdot}{\mathbin{\mathaccent\cdot\cup}}

\newcommand{\Attr}{\mathit{Attr}}

\newcommand{\pr}{\textsf{pr}}

\newcommand{\invalpha}{{\overline{\alpha}}}

\newcommand{\sqdiamond}{\tikz [x=1.2ex,y=1.2ex,line width=.08ex] \draw (0,.5) -- (.5,1) -- (1,.5) -- (.5,0) -- (0,.5) -- cycle;}
\newcommand{\sqsq}{\tikz [x=0.95ex,y=1ex,line width=.1ex] \draw (0,0) -- (1,0) -- (1,1) -- (0,1) -- (0,0) -- cycle;}

\newcommand{\Even}{{\raisebox{0.13ex}{\scalebox{0.95}{\sqdiamond}}}}
\newcommand{\Odd}{{\raisebox{0.15ex}{\scalebox{0.9}{$\sqsq$}}}}

\newcommand{\Veven}{V_{\Even}}
\newcommand{\Vodd}{V_{\Odd}}

\DeclareRobustCommand{\Vis}{\textsf{Visit}}
\DeclareRobustCommand{\Avoid}{\textsf{Avoid}}
\DeclareRobustCommand{\bdiamond}{\mathbin{\textbf{\raisebox{.25ex}{\scalebox{.6}{\rotatebox[origin=c]{45}{$\Box$}$\,$}}}}}
\DeclareRobustCommand{\bbox}{\mathbin{\textbf{\raisebox{.25ex}{\scalebox{.6}{{$\Box$}$\:$}}}}}

\mdfdefinestyle{todo}{tikzsetting={draw=green,fill=green!10,line width=3pt},leftmargin=0pt,
	rightmargin=0pt,innertopmargin=6pt,innerbottommargin=6pt,innerleftmargin=6pt,innerrightmargin=6pt}
\newenvironment{TodoBoxedInternal}[1][]
{\begin{mdframed}[style=todo]\parindent0pt}{\end{mdframed}}

\iftrue
\newcommand{\Todo}[1]{%
	\begin{TodoBoxedInternal}
		#1
	\end{TodoBoxedInternal}%
}
\else
\newcommand{\Todo}[1]{}
\fi

\title{Simple Fixpoint Iteration To Solve Parity Games}
\def\titlerunning{Simple Fixpoint Iteration To Solve Parity Games}

\author{
  Tom van Dijk\
  \institute{Formal Methods and Tools \\ University of Twente, Enschede}
  \email{t.vandijk@utwente.nl}
  \and Bob Rubbens
  \institute{Formal Methods and Tools \\ University of Twente, Enschede}
  \email{r.b.rubbens@student.utwente.nl}
}
\def\authorrunning{Tom van Dijk \& Bob Rubbens}

\begin{document}
	
\maketitle

\begin{abstract}

A naive way to solve the model-checking problem of the mu-calculus uses fixpoint iteration.
Traditionally however mu-calculus model-checking is solved by a reduction in linear time to a parity game, which is then solved using one of the many algorithms for parity games.

We now consider a method of solving parity games by means of a naive fixpoint iteration.
Several fixpoint algorithms for parity games have been proposed in the literature.
In this work, we introduce an algorithm that relies on the notion of a distraction.
The idea is that this offers a novel perspective for understanding parity games.
We then show that this algorithm is in fact identical to two earlier published fixpoint algorithms for parity games and thus that these earlier algorithms are the same.

Furthermore, we modify our algorithm to only partially recompute deeper fixpoints after updating a higher set and show that this modification enables a simple method to obtain winning strategies.

We show that the resulting algorithm is simple to implement and offers good performance on practical parity games.
We empirically demonstrate this using games derived from model-checking, equivalence checking and reactive synthesis and show that our fixpoint algorithm is the fastest solution for model-checking games.

\end{abstract}

\section{Introduction}
\label{sec:introduction}

Parity games are turn-based games played on a finite directed graph.
Two players \emph{Odd} and \emph{Even} move a token along the edges of the graph, yielding an infinite play. 
Each vertex belongs to exactly one player, who decides the successor vertex in the play.
Vertices are labeled with a natural number \emph{priority}.
The winner of a play is determined by the highest priority that is encountered infinitely often along the play.
Player Odd wins if this priority is odd; otherwise, player Even wins.

Parity games are interesting both for their practical applications and for complexity theoretic reasons.
Their study has been motivated by their relation to many problems in formal verification and synthesis that can be reduced to the problem of solving parity games, as many properties of programs are naturally specified by means of fixpoints and parity games capture the expressive power of nested least and greatest fixpoint operators.
In particular, there is a tight connection with the modal $\mu$-calculus~\cite{DBLP:reference/mc/BradfieldW18,DBLP:journals/tcs/Kozen83}. The verification and satisfiability problems of the $\mu$-calculus can be linearly reduced to deciding the winner of a parity game~\cite{DBLP:journals/tcs/EmersonJS01,DBLP:conf/concur/Stirling95} and solving parity games can be linearly reduced to a formula in the modal $\mu$-calculus~\cite{DBLP:conf/stacs/Walukiewicz96}.

Parity games are interesting for complexity theory, 
as the problem of determining the winner of a parity game is known to lie in $\text{UP}\cap\text{co-UP}$~\cite{DBLP:journals/ipl/Jurdzinski98},
%
which is contained in $\text{NP}\cap\text{co-NP}$~\cite{DBLP:journals/tcs/EmersonJS01}. 
The problem is therefore unlikely to be NP-complete and it is widely believed that a polynomial solution exists.
Earlier subexponential and recent quasi-polynomial solutions to parity games strengthen this belief.
Despite much effort, a polynomial-time algorithm has not been found yet.



It seems evident that we require a better understanding of parity games to answer whether they can be solved in polynomial time.
This paper is part of an effort to understand parity games by considering how different algorithms deal with so-called \emph{distractions} and to see if perhaps we can combine features from different algorithms.
As a direct result of this effort, we have discovered a very easy method to obtain winning strategies for the possibly most naive solution to parity games, which is via fixpoint iteration.
The method exposes the relationship between fixpoint iteration and the famous recursive algorithm by Zielonka.
For various treatments of Zielonka's recursive algorithm, we refer to~\cite{DBLP:conf/tacas/Dijk18,Verver2013,DBLP:journals/tcs/Zielonka98}

Solutions to parity games via fixpoint computation essentially translate the game into a formula of the $\mu$-calculus which is then solved naively.
Two such algorithms have been proposed in the literature.
Based on the formulas by Walukiewicz~\cite{DBLP:conf/stacs/Walukiewicz96} that translate the winning condition of a parity game to $\mu$-calculus formulas over the parity game, Bruse et al. proposed a fixpoint algorithm we call \texttt{BFL}~\cite{DBLP:journals/corr/BruseFL14}.
Based on earlier work by Kupferman and Vardi~\cite{DBLP:conf/stoc/KupfermanV98}, Di Stasio et al. implemented the
\texttt{APT}~\cite{DBLP:conf/wia/StasioMPV16} algorithm.

The contributions of this paper are the following.
We \textbf{discuss} distractions in parity games and how they offer a compelling perspective to study different features of parity game solving algorithms.
We \textbf{present} a novel distraction fixpoint iteration (\texttt{DFI}) algorithm based on computing the distractions in a parity game and \textbf{prove} its correctness based on a construction of winning strategies.
We \textbf{propose} an optimization to \texttt{DFI} that only partially recomputes the lower fixpoints and show how this optimization trivializes strategy computation.
We \textbf{demonstrate} that \texttt{DFI} is {simple to implement}.
We \textbf{compare} \texttt{DFI} to the two fixpoint algorithms \texttt{APT} and \texttt{BFL} and show that all three are equivalent.
We \textbf{show} empirically that the \texttt{DFI} algorithm is efficient for practical parity games and the fastest solution for practical model-checking games.

\section{Preliminaries}
\label{sec:preliminaries}

\subsection{Parity games}


We formally define a parity game $\Game$ as a tuple $(\Veven, \Vodd, E, \pr)$ where $V=\Veven\cupdot \Vodd$ is a set of $n$ vertices partitioned into the sets $\Veven$ controlled by player \emph{Even} and $\Vodd$ controlled by player \emph{Odd}, and $E\subseteq V\times V$ is a left-total binary relation describing all moves. 
Every vertex has at least one successor.
We also write $E(u)$ for all successors of $u$ and $u\rightarrow v$ for $v\in E(u)$.
The function $\pr\colon V\rightarrow \{0,1,\dotsc,d\}$
assigns to each vertex a \emph{priority}, where $d$ is the highest priority in the game.
We write $\alpha\in\{\Even,\Odd\}$ to denote a player $\Even$ or $\Odd$ and $\invalpha$ for the opponent of $\alpha$ and we also use $0$ for player Even and 1 for player Odd.
When representing a parity game visually, we use diamonds for vertices of player Even and boxes for vertices of player Odd.


We write
$\pr(V)$ for the highest priority of vertices $V$ and
$\pr(\Game)$ for the highest priority in the game $\Game$.
We write $V_p$ for the set of vertices with priority $p$. 
With $V_\text{even}$ and $V_\text{odd}$ we denote all vertices with an even or odd priority.
Notice that $V_\Even$ and $V_\text{even}$ are not the same sets and the same holds for $V_\Odd$ and $V_\text{odd}$.

A \emph{play} $\pi=v_0 v_1 \dots$ is an infinite sequence of vertices consistent with $E$, i.e.,
$v_i \rightarrow v_{i+1}$ for all successive vertices.
We denote with $\inf(\pi)$ the vertices that occur infinitely often in $\pi$.
Player Even wins a play $\pi$ if $\pr(\inf(\pi))$ is even; player Odd if $\pr(\inf(\pi))$ is odd.
%

A (positional) \emph{strategy} $\sigma\subseteq V\to V$ assigns to each vertex in its domain a single successor in $E$, i.e., $\sigma\subseteq E$.
We refer to a strategy of player $\alpha$ to restrict the domain of $\sigma$ to $V_\alpha$. 
In the remainder, all strategies $\sigma$ are of a player $\alpha$.
We write $\text{Plays}(v)$ for the set of plays starting at vertex $v$.
We write $\text{Plays}(v, \sigma)$ for all plays from $v$ consistent with $\sigma$,
and $\text{Plays}(V,\sigma)$ for $\{\,\pi\in\text{Plays}(v,\sigma)\mid v\in V\,\}$.

A basic result for parity games is that they are memoryless determined~\cite{DBLP:conf/focs/EmersonJ91}, i.e., each vertex is either winning for player Even or for player Odd, and both players have a strategy for their winning vertices.
Player $\alpha$ wins vertex $v$ if they have a strategy $\sigma$ such that every $\pi\in\text{Plays}(v,\sigma)$ is winning for player $\alpha$.



\iftrue
Several algorithms for solving parity games employ \emph{attractor computation}.
Given a set of vertices $A$,
the attractor of $A$ for a player $\alpha$ represents those vertices from which player $\alpha$ can force a play to visit $A$.
We write $\Attr^\Game_\alpha(A)$ to attract vertices in $\Game$ to $A$ as player $\alpha$,
i.e.,
the fixpoint of
\[Z := A \cup \{\;v\in V_\alpha \mid E(v)\cap Z \neq \emptyset\;\} \cup \{\;v\in V_{\invalpha} \mid E(v)\subseteq Z\;\}\]
Informally, we compute the $\alpha$-attractor of $A$ with a backward search from $A$, initially setting $Z:=A$ and iteratively adding $\alpha$-vertices with a successor in $Z$ and $\invalpha$-vertices with no successors outside $Z$.
We call a set of vertices $A$ $\alpha$-maximal if $A=\Attr^\Game_\alpha(A)$. 
The attractor also yields an ``attractor strategy'' by selecting a vertex in $Z$ for every added $\alpha$-vertex $v$ when $v$ is added to $Z$, and by selecting a vertex in $Z$ for all $\alpha$-vertices in $A$ that do not yet have a strategy but can play to $Z$.
\fi

\subsection{The modal \texorpdfstring{$\mu$}{mu}-calculus}

We now introduce the modal $\mu$-calculus~\cite{DBLP:reference/mc/BradfieldW18,DBLP:journals/tcs/Kozen83},
which we use in this paper as an aid to describe and reason about the algorithms.
%
We define the $\mu$-calculus over parity games.
Formulas are constructed using conjunction, disjunction, modalities and fixpoint operators.
Let $\mathbb{X}$ be a set of \emph{second-order variables}.
The set of $\mu$-calculus formulas in positive normal form is generated by the following grammar:
\[
\phi ::= U \mid \neg U \mid \phi\land\phi \mid \phi\lor\phi \mid \bdiamond\phi \mid \bbox\phi \mid \mu X.\phi \mid \nu X.\phi \mid X
\]
where $X\in\mathbb{X}$ and $U\subseteq V$ is a set of vertices of the parity game,
for example the sets $V_\Even$, $V_\Odd$, $V_p$, $V_\text{even}$ and $V_\text{odd}$ introduced above.
Given some parity game $\Game$, an interpretation of the variables in $\mathbb{X}$ is a mapping $\rho\colon \mathbb{X}\to 2^V$.
The semantics of this $\mu$-calculus is a set of vertices, inductively defined as follows.
\[
\begin{tabu}{lrl}
\llbracket U \rrbracket_\rho & := & U \\
\llbracket \neg U \rrbracket_\rho & := & V\setminus U \\
\llbracket \phi\land\psi \rrbracket_\rho & := & \llbracket \phi \rrbracket_\rho \cap \llbracket \psi \rrbracket_\rho\\
\llbracket \phi\lor\psi \rrbracket_\rho & := & \llbracket \phi \rrbracket_\rho \cup \llbracket \psi \rrbracket_\rho\\
\llbracket \bdiamond\phi \rrbracket_\rho & := & \{ v\in V \mid \exists u\in E(v). u\in\llbracket \phi \rrbracket_\rho \} \\
\llbracket \bbox\phi \rrbracket_\rho & := & \{ v\in V \mid \forall u\in E(v). u\in\llbracket \phi \rrbracket_\rho \} \\
\llbracket \mu X.\phi \rrbracket_\rho & := & \bigcap \{ U\subseteq V \mid \llbracket \phi \rrbracket_{\rho[X\mapsto U]}\subseteq U \} \\
& = & \textit{lfp}(U := \llbracket \phi \rrbracket_{\rho[X\mapsto U]}) \\
\llbracket \nu X.\phi \rrbracket_\rho & := & \bigcup \{ U\subseteq V \mid U \subseteq \llbracket \phi \rrbracket_{\rho[X\mapsto U]} \} \\
& = & \textit{gfp}(U := \llbracket \phi \rrbracket_{\rho[X\mapsto U]}) \\
\llbracket X \rrbracket_\rho & := & \rho(X) \\
\end{tabu}
\]

It is well-known that the semantics of $\mu X.\phi$ is the \emph{least fixpoint} of $\phi$: starting with $\textbf{X}:=\emptyset$, we compute $\mathbf{X} := \llbracket \phi \rrbracket_{\rho[X\mapsto \mathbf{X}]}$ until \textbf{X} is a fixpoint; similarly, we compute the semantics of $\nu X.\phi$ with the \emph{greatest fixpoint} of $\phi$, starting with $\mathbf{X}:=V$, we apply $\phi$ until \textbf{X} is a fixpoint.


\section{Distractions}
\label{sec:distractions}

Imagine a parity game for which we know the winning regions and now we want to compute the winning strategies for both players.
The most naive method would select a random successor inside the winning region for every vertex of the winning player.
This is obviously incorrect as the following example shows.
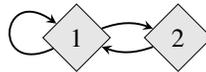
\begin{center}
	\scalebox{0.9}{
	\begin{tikzpicture}		
		\tikzset{every edge/.append style={>=stealth,->,solid,thick,draw,text height=0.5ex,text depth=0.2ex}}
		\tikzset{every node/.append style={inner sep=0,minimum size=10mm,draw,fill=black!10}}
		\tikzset{my label/.style args={#1:#2}{
			append after command={
				($(\tikzlastnode.center)$) coordinate [label={[label distance=5mm,black]#1:\textbf{\strut #2}}]
		}}}
		\tikzstyle{even}=[diamond]
		\tikzstyle{odd}=[regular polygon,regular polygon sides=4]

	\draw (0,0)   node[even] (a) {1};
	\draw (1.5,0) node[even] (b) {2};
	\draw (a) edge [in=145,out=-155,loop,looseness=7] (a);
	\draw (a) edge [bend left=20] (b);
	\draw (b) edge [bend left=20] (a);
	\end{tikzpicture}
}
\end{center}
In this example player Even controls both vertices and wins the entire game by playing from \textbf{1} to \textbf{2},
as this results in each play alternating between the two vertices, i.e., the highest priority seen infinitely often is \textit{even}.
If however player Even would select a random strategy and select as the strategy to play from \textbf{1} to \textbf{1},
then the highest priority seen infinitely often would be \textit{odd} and player Even would lose.
Selecting a random strategy that stays in the won region is therefore incorrect.

A smarter but still incorrect method is to try to play to the highest priority vertex of the winner's parity, by repeatedly attracting to the highest vertices like in attractor-based algorithms~\cite{DBLP:conf/cav/BenerecettiDM16,DBLP:conf/cav/Dijk18,DBLP:journals/tcs/Zielonka98}.
This however does not always produce a correct result.
Sometimes playing to an ``attractive high priority'' vertex can be a deception, as the following example shows.

%
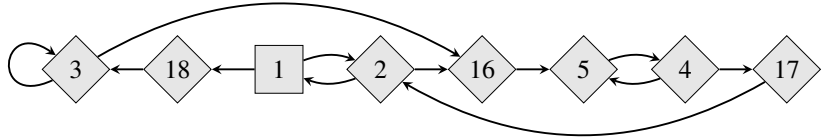
\begin{center}
	\scalebox{0.9}{
		\begin{tikzpicture}		
		\tikzset{every edge/.append style={>=stealth,->,solid,thick,draw,text height=0.5ex,text depth=0.2ex}}
		\tikzset{every node/.append style={inner sep=0,minimum size=10mm,draw,fill=black!10}}
		\tikzset{my label/.style args={#1:#2}{
				append after command={
					($(\tikzlastnode.center)$) coordinate [label={[label distance=5mm,black]#1:\textbf{\strut #2}}]
		}}}
		\tikzstyle{even}=[diamond]
		\tikzstyle{odd}=[regular polygon,regular polygon sides=4]
		
		\draw (-4.5,0) node[even]  (y) {3};
		\draw (-3.0,0) node[even]  (x) {18};
		\draw (-1.5,0) node[odd ]  (a) {1};
		\draw (0,0)    node[even]  (b) {2};
		\draw (1.5,0)  node[even]  (c) {16};
		\draw (3.0,0)  node[even]  (d) {5};
		\draw (4.5,0)  node[even]  (e) {4};
		\draw (6.0,0)  node[even]  (f) {17};
		\draw (y) edge [in=145,out=-155,loop,looseness=7] (y);
		\draw (y) edge [bend left=30] (c);
		\draw (x) edge (y);
		\draw (a) edge (x);
		\draw (a) edge [bend left=20] (b);
		\draw (b) edge [bend left=20] (a);
		\draw (b) edge (c);
		\draw (c) edge (d);
		\draw (d) edge [bend left=20] (e);
		\draw (e) edge [bend left=20] (d);
		\draw (e) edge (f);
		\draw (f) edge [bend left=30] (b);
		\end{tikzpicture}
	}
\end{center}
This entire game is won by player Even.
However, picking a winning strategy is not trivial.
The result of attractor computation is that no vertex is attracted to \textbf{18},
no vertex is attracted to \textbf{17},
vertices $\{ \textbf{1}, \textbf{2}, \textbf{3} \}$ are attracted to \textbf{16},
and finally vertex \textbf{4} is attracted to \textbf{5}.
If player Even plays from \textbf{3} to \textbf{3}, they lose.
Playing from \textbf{4} to \textbf{5} is losing.
Less obvious is that playing from \textbf{2} to \textbf{16} is also losing,
because player Even must play via \textbf{5} and \textbf{4} to \textbf{17}.
Any play that sees \textbf{16} infinitely often also sees \textbf{17} infinitely often.
%
%
Player Even can only win by playing from \textbf{18} to \textbf{3} to \textbf{16} to \textbf{5} to \textbf{4} to \textbf{17} to \textbf{2} to \textbf{1}.
Then player Odd can either choose to play from \textbf{1} to \textbf{2} or to \textbf{18}, i.e., to a cycle with priority 2 or to a cycle with priority 18.
%

In order to win, player Even must \emph{not} play to \textbf{16} from some vertices.
%
%
%
We propose to call vertices like vertex \textbf{16} {distractions}.
A \textbf{distraction} for player $\alpha$ is a vertex $v$ with an $\alpha$-priority $p$, such that if player $\alpha$ always plays to reach $v$ along paths of priorities $\leq p$, then player $\invalpha$ wins $v$ and all vertices that reach $v$.
That is, a distraction for $\alpha$ is a high value vertex $v$ with an $\alpha$-priority that player $\invalpha$ can win if player $\alpha$ always tries to visit it.
%
This occurs either when player $\invalpha$ can attract $v$ to vertices with higher priorities of $\invalpha$, so every cycle with $v$ also visits one of these vertices, or when player $\invalpha$ can attract $v$ to a (lower) $\invalpha$-dominion.
We distinguish clear distractions and devious distractions.
A \textbf{clear distraction} is a vertex $v$ that is a distraction for player $\alpha$ and in the winning region of player $\invalpha$.
While solving the game, player $\alpha$ may initially try to visit $v$, but at some point the algorithm determines that $v$ is losing for player $\alpha$ and player $\alpha$ then avoids $v$.
A \textbf{devious distraction} is a vertex $v$ that is a distraction for player $\alpha$ and in the winning region of player $\alpha$.
That is, player $\alpha$ wins vertex $v$ but only by not playing towards $v$ from some vertices that could play to $v$. 
We conjecture that whenever simple attractor computation yields an incorrect strategy, this is due to devious distractions.

If a vertex $v$ is a distraction because the opponent can attract $v$ to higher vertices of $\invalpha$'s priorities, 
then this is either \emph{trivial}, when player $\invalpha$ directly attracts $v$, or player $\invalpha$ attracts $v$ via a \emph{tangle}.
%
A \textbf{tangle}~\cite{DBLP:conf/cav/Dijk18} is a subgame where one player has a strategy to win all plays that stay inside the tangle.
The other player must therefore escape the tangle.
These escapes then lead to vertices of higher priorities of player $\invalpha$, ensuring that any play that visits $v$ infinitely often also visits one of these vertices infinitely often.
%
%
%
In the example, \textbf{17} does not directly attract \textbf{16}, but \textbf{4} and \textbf{5} form a tangle which player Even must escape.

The vertices with a higher $\invalpha$'s priority that player $\invalpha$ can force plays from a devious distraction $v$ to, are themselves \emph{clear} distractions (for player $\invalpha$) that are won by player $\alpha$, otherwise player $\alpha$ would not win vertex $v$.
Then, when a parity game solver decides that these vertices are distractions for $\invalpha$ and won by $\alpha$, the strategy that is a witness to this is also the strategy that avoids the devious distraction $v$, otherwise $v$ would not be a devious distraction.
Thus, simply remembering the strategy used to decide clear distractions yields the correct strategy to avoid devious distractions.

We now consider briefly how different parity game solving algorithms deal with distractions.
Nearly every solver initially prefers to play to the highest priority.
A solver must at some point decide that a nice high priority vertex is not actually a desirable target.
The challenge is thus to recognize early that some vertices are distractions and the key question to ask of these algorithms is how they accomplish this.
We believe that this may be a key to a deeper understanding of parity games.


For the attractor-based algorithms like Zielonka's recursive algorithm~\cite{DBLP:journals/tcs/Zielonka98}, priority promotion~\cite{DBLP:conf/cav/BenerecettiDM16} and tangle learning~\cite{DBLP:conf/cav/Dijk18},
reaching a vertex with the highest priority remains the goal until it is actually attracted by the opponent.
For algorithms that employ progress measures, such as small progress measures~\cite{DBLP:conf/stacs/Jurdzinski00}, often reaching high priority vertices is the primary goal, but vertices with $\alpha$'s priority along the path get a higher value; thus if the value of a distraction does not increase, the algorithm lifts vertices along alternative paths, basically \emph{ignoring} distractions.
This is especially true with the recent quasipolynomial solutions, the ``succinct progress measures''~\cite{DBLP:conf/lics/JurdzinskiL17} and the ``ordered progress measures''~\cite{DBLP:journals/sttt/FearnleyJKSSW19}, as the value of vertices with progress quickly overtakes vertices with only a high priority.

Algorithms based on progress measures cannot find that distractions are attracted by the opponent. They are fundamentally unable to do so because they only compute progress measures from the perspective of a single player.
Current attractor-based algorithms have no mechanism to ignore vertices that have no good continuation, because they do not explicitly ignore such vertices or assign a higher value to vertices along the path to a high priority vertex.
No current algorithm combines these features.

%



An open question is whether parity games without devious distractions are easier to solve than parity games with devious distractions.
Even clear distractions can make many algorithms slow, exponentially so, as we demonstrate for attractor-based algorithms with the Two Counters game~\cite{DBLP:journals/corr/abs-1807-10210}.
Notice also that the above discussion implies a partition of every winning region into subgames separated by the clear distractions where the strategy inside each subgame is the witness strategy that avoids any devious distractions (see e.g. Fig.~\ref{fig:regions} below).

\section{The distraction fixpoint iteration (\texttt{DFI}) algorithm}
\label{sec:dfi}
We propose an algorithm that does not directly compute the winning regions,
but instead computes which vertices are distractions.
This leads to an algorithm that naturally follows the intuition to
assume first that all vertices are won by the player of the parity of their priority,
and then to refine this estimation.

\subsection{Computing the distractions by means of fixpoints}

The algorithm maintains a series of sets $Z_0, Z_1, \dotsc, Z_d$ for the priorities $0,1,\dotsc,d$ in the game.
These sets are updated by a nested fixpoint operation and in the final state contain all vertices that are not won by the player of the parity of their priority.
Each set $Z_p$ contains the vertices that are estimated to be won by player $1-(p \bmod 2)$.
For example, vertices in $Z_5$ are estimated won by player Even.
$Z_p$ thus identifies distractions while solving the game,
even though in the final state of the computation, distractions that are won by their player, i.e., devious distractions, are not in the $Z_p$ sets.
We use the $Z_p$ sets such that membership in $Z_p$ is only relevant for vertices in $V_p$,
by working with a set $Z$ such that
\[
\begin{tabu}{rrl}
Z & := & (V_0 \land Z_0) \lor (V_1 \land Z_1) \lor \dots \lor (V_d \land Z_d) \\
  &  = & \bigvee_{p=0}^d (V_p\land Z_p)
\end{tabu}
\]

We compute who wins each vertex according to this set $Z$,
\[
\begin{tabu}{rrl}
\textsf{winner}(v, Z) & := &
	 \begin{cases} \pr(v) \bmod 2 & v\notin Z \\ 1 - (\pr(v) \bmod 2) & v \in Z \end{cases} \\
\textsf{Even}(Z) & := & \{ v \mid \textsf{winner}(v, Z) = 0 \} \\
\textsf{Odd}(Z) & := & \{ v \mid \textsf{winner}(v, Z) = 1 \} \\
\end{tabu}
\]

We can equivalently define these sets using the $\mu$-calculus notation,
\[
\begin{tabu}{rrl}
\textsf{Even}(Z) & := & (V_\text{even} \land \neg Z) \vee (V_\text{odd} \land Z) \\
\textsf{Odd}(Z) & := & (V_\text{even} \land Z) \vee (V_\text{odd} \land \neg Z) \\
\end{tabu}
\]

The fundamental idea of the algorithm is to estimate whether vertices are distractions based only on the direct successors and to update this estimate in a strict order, beginning with the least important vertices, and resetting the estimates of lower vertices whenever a higher vertex is updated.
To estimate whether a vertex is won in one step, based on its direct successors, given $Z$, we compute
\[
\begin{tabu}{rrl}
\textsf{onestep}(v, Z) & := & 
\begin{cases}
0 & \quad v \in V_\Even \land \exists u \in E(v) . \textsf{winner}(u, Z) = 0\\ 
1 & \quad v \in V_\Even \land \forall u \in E(v) . \textsf{winner}(u, Z) = 1\\ 
1 & \quad v \in V_\Odd \land \exists u \in E(v) . \textsf{winner}(u, Z) = 1\\ 
0 & \quad v \in V_\Odd \land \forall u \in E(v) . \textsf{winner}(u, Z) = 0\\ 
\end{cases} \\
\textsf{Onestep}_0(Z) & := & \{v\mid\textsf{onestep}(v,Z) = 0\} \\
\textsf{Onestep}_1(Z) & := & \{v\mid\textsf{onestep}(v,Z) = 1\} \\
\end{tabu}
\]

Or defined equivalently using the $\mu$-calculus notation,
\[
\begin{tabu}{rrl}
\textsf{Onestep}_0(Z) & := & (V_\Even \land \bdiamond \textsf{Even}(Z) ) \lor (V_\Odd \land \bbox \textsf{Even}(Z)) \\
\textsf{Onestep}_1(Z) & := & (V_\Even \land \bbox \textsf{Odd}(Z) ) \lor (V_\Odd \land \bdiamond \textsf{Odd}(Z)) \\
\end{tabu}
\]

Now that we can compute whether a vertex is forced in one step to the (estimated) winning region of each player, we can easily compute all vertices that are distractions according to this estimation,
\[
\begin{tabu}{rrl}
\textsf{OnestepDistraction}(Z) & := & (V_\text{even} \land \textsf{Onestep}_1(Z)) \lor (V_\text{odd} \land \textsf{Onestep}_0(Z)) \\
\end{tabu}
\]

These are all the vertices with even priorities that are estimated to be won in one step by player Odd and vice versa.
In a way, we are checking whether the estimated winning regions are closed, and update the vertices that should actually belong to the winning region of the other player.
We now solve parity games by computing which vertices are distractions,
\[
\begin{tabu}{rrl}
\textsf{Distraction} & := & \mu Z_d \dots \mu Z_1\ .\ \mu Z_0\ .\ \textsf{OnestepDistraction}\big(\bigvee_{p=0}^d (V_p\land Z_p)\big)) \\
\end{tabu}
\]

Having computed all distractions, obtaining the winning regions is trivial.
Player Even wins $\textsf{Even}(\textsf{Distraction})$ and player Odd wins $\textsf{Odd}(\textsf{Distraction})$.
Notice that due to how $\mu$-calculus formulas work, all sets $Z_0\dots Z_d$ contain all (simple) distractions, however in practice we only consider vertices with priority $p$ when updating the set $Z_p$.
This optimization is important and has been noted by~\cite{DBLP:journals/corr/BruseFL14}.

We can now implement the basic algorithm as in Algorithm~\ref{alg:dfi1}.
We simply begin with the lowest vertices and work our way up until a set $Z_p$ is changed, upon which we reset all lower sets.

\begin{algorithm}[tb]
	\Def{\DFI{$\Game$}}{
		$Z \leftarrow \emptyset$ \tcp*{start with no distractions}
		$p \leftarrow 0$ \tcp*{start with lowest priority}
		\While(\tcp*[f]{while $\leq$ highest priority}){$p\leq d$}{
			$\alpha \leftarrow p \bmod 2$ \tcp*{current parity}
			$Y \leftarrow \{ v \in V_p \setminus Z \mid \textsf{onestep}(v, Z)\neq \alpha \}$ \tcp*{new distractions}
			\If{$Y\neq\emptyset$}{
				$Z \leftarrow Z\cup Y$ \tcp*{update current fixpoint $Z_p$}
				$Z \leftarrow Z \setminus \{ v \mid \pr(v) < p\}$ \tcp*{reset all lower fixpoints}
				$p \leftarrow 0$ \tcp*{restart with lowest priority}
			}
			\Else{
				$p \leftarrow p + 1$ \tcp*{fixpoint, continue higher}
			}
		}
		\Return $W_\Even, W_\Odd$ where $W_\Even$ $\leftarrow$ $\{v\mid \textsf{winner}(v,Z)=0\}$, $W_\Odd$ $\leftarrow$ $V \setminus W_\Even$ 
	}
\caption{The basic \texttt{DFI} algorithm}
\label{alg:dfi1}
\end{algorithm}

\subsection{Proving \texttt{DFI} correct by constructing the winning strategy}

In the following, we prove that \texttt{DFI} correctly computes the winning regions.
We prove this by constructing the winning strategies for both players in their won regions.
Notice that the \texttt{DFI} algorithm as presented above does not actually compute this winning strategy explicitly.

In the following, we define the set of vertices $V_{\leq p} := \{ v \mid \pr(v) \leq p\}$ and also the set $\textsf{Won}(\alpha)_{\leq p} := \{ v \in V_{\leq p} \mid \textsf{winner}(v,Z) = \alpha \}$ of vertices in $V_{\leq p}$ won by player $\alpha$ according to estimation $Z$.

\begin{lemma}
After computing the fixpoint of $Z_p$,
player $\alpha\in\{\Even,\Odd\}$ has a winning region $W_\alpha\equiv\textsf{Won}(\alpha)_{\leq p}$ and a strategy $\sigma_\alpha$ for all $v\in V_\alpha\cap W_\alpha$,
such that 
\begin{itemize}[nosep]
	\item $\alpha$ never plays from $W_\alpha$ to $\textsf{Won}(\invalpha)$
	\item $\invalpha$ cannot play from $W_\alpha$ to $\textsf{Won}(\invalpha)$
	\item all cycles consistent with $\sigma_\alpha$ in $W_\alpha$ are won by $\alpha$
\end{itemize}
\label{lemma:str}
\end{lemma}

\begin{proof}
We prove by induction.
Lemma~\ref{lemma:str} is trivially true for the empty game.
We assume that Lemma~\ref{lemma:str} holds 
after computing the fixpoint of $Z_{p-1}$ and show that it holds 
after computing the fixpoint of $Z_p$.
To prove the property for both players, we distinguish the two cases $\alpha \neq (p \bmod 2)$ and $\alpha = (p \bmod 2)$.

\textbf{Case 1. $\alpha \neq (p \bmod 2)$}.\qquad

We maintain a set of won vertices $W_\alpha\subseteq V_{\leq p}$ and a strategy $\sigma_\alpha$ for all vertices in $V_\alpha\cap W_\alpha$.
We update $W_\alpha$ and $\sigma_\alpha$ after every iteration of the fixpoint computation.
We prove that Lemma~\ref{lemma:str} holds for $W_\alpha$ and $\sigma_\alpha$ after each update.
Initially, $W_\alpha=\emptyset$ and $\sigma_\alpha=\emptyset$ so Lemma~\ref{lemma:str} trivially holds.
Recall that the fixpoint operation repeatedly recomputes the fixpoint of $Z_{p-1}$ and updates $Z_p$ with all $p$-vertices that now play in one step to $\textsf{Won}(\alpha)$, the winning region of $\alpha$ in the entire game.
After computing the lower fixpoint $Z_{p-1}$, we obtain the set $\textsf{Won}(\alpha)_{<p}$ and the strategy $\sigma_{\alpha, <p}$ (for the vertices $\alpha$ wins in $V_{<p}$).
We then set
\[
\begin{tabu}{rrl}
W'_\alpha & := & W_\alpha \cup \textsf{Won}(\alpha)_{<p} \\
\sigma'_\alpha & := & \sigma_\alpha \cup \big(\sigma_{\alpha,<p}\cap(\textsf{Won}(\alpha)\setminus W_\alpha)\big)
\end{tabu}
\]
where $W'_\alpha$ and $\sigma'_\alpha$ will be the next $W_\alpha$ and $\sigma_\alpha$.
It is critical that we keep the old strategy $\sigma_\alpha$ for all vertices in $W_\alpha$.
Notice that the region won by $\alpha$ in the lower game monotonically increases with every iteration, since the only difference between the fixpoint iterations is that more vertices with priority $p$ are now won by $\alpha$.
That is, all vertices in $W_\alpha$ are still won after recomputing the fixpoint of $Z_{p-1}$.
We then add each vertex in $V_p \cap \textsf{Onestep}_\alpha$ to $W'_\alpha$ and exactly when we add such a vertex controlled by $\alpha$ to $W'_\alpha$, we choose a successor in $\textsf{Won}(\alpha)$ as the strategy for that vertex to $\sigma'_\alpha$.
Now $W'_\alpha$ contains exactly all vertices in $\textsf{Won}(\alpha)_{\leq p}$ after the update to $Z_p$
and $\sigma'_\alpha$ has a strategy for all vertices in $V_\alpha\cap W'_\alpha$.

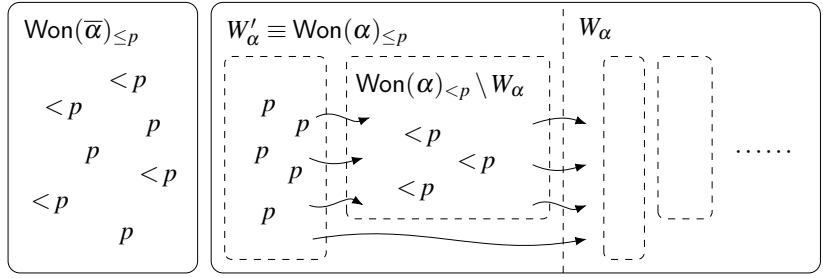
\begin{figure}[tbp]
\begin{center}
\scalebox{0.9}{
\begin{tikzpicture}


\draw[rounded corners=5pt] (-3,0) { -- ++(0,4) -- ++(2.8,0)} -- ++(0,-4) -- cycle;
\draw[rounded corners=5pt] (0,0) { -- ++(0,4) -- ++(9,0)} -- ++(0,-4) -- cycle;
\draw[rounded corners=3pt,dashed] (0.2,0.2) { -- ++(0,3) -- ++(1.5,0)} -- ++(0,-3) -- cycle;
\draw[rounded corners=3pt,dashed] (2.0,0.8) { -- ++(0,2.4) -- ++(3.0,0)} -- ++(0,-2.4) -- cycle;
\draw[dashed] (5.2,0) -- ++(0,4);
\draw[rounded corners=3pt,dashed] (5.8,0.2) { -- ++(0,3) -- ++(0.6,0)} -- ++(0,-3) -- cycle;
\draw[rounded corners=3pt,dashed] (6.6,0.8) { -- ++(0,2.4) -- ++(0.8,0)} -- ++(0,-2.4) -- cycle;

\draw node[align=left,anchor=north west] at (-2.9,3.9) {$\textsf{Won}(\invalpha)_{\leq p}$};

\draw node[align=left,anchor=north west] at (0.1,3.9) {$W'_\alpha\equiv\textsf{Won}(\alpha)_{\leq p}$};

\draw node[draw=none,align=left,anchor=north west] at (0.6,1.1) {$p$};
\draw node[draw=none,align=left,anchor=north west] at (1.0,1.7) {$p$};
\draw node[draw=none,align=left,anchor=north west] at (0.5,2) {$p$};
\draw node[draw=none,align=left,anchor=north west] at (1.1,2.4) {$p$};
\draw node[draw=none,align=left,anchor=north west] at (0.6,2.7) {$p$};

\draw node[draw=none,align=left,anchor=north west] at (2.0,3.1) {$\textsf{Won}(\alpha)_{<p}\setminus W_\alpha$};

\draw node[draw=none,align=left,anchor=north west] at (5.3,3.9) {$W_\alpha$};

\draw node[draw=none,align=left,anchor=north west] at (7.6,2) {$\cdots\cdots$};

\draw node[draw=none,align=left,anchor=north west] at (-2.8,1.3) {$<\!p$};
\draw node[draw=none,align=left,anchor=north west] at (-1.2,1.7) {$<\!p$};
\draw node[draw=none,align=left,anchor=north west] at (-2.0,2) {$p$};
\draw node[draw=none,align=left,anchor=north west] at (-1.1,2.4) {$p$};
\draw node[draw=none,align=left,anchor=north west] at (-2.6,2.7) {$<\!p$};
\draw node[draw=none,align=left,anchor=north west] at (-1.5,0.8) {$p$};
\draw node[draw=none,align=left,anchor=north west] at (-1.65,3.1) {$<\!p$};

\draw node[draw=none,align=left,anchor=north west] at (2.6,1.5) {$<\!p$};
\draw node[draw=none,align=left,anchor=north west] at (3.5,1.9) {$<\!p$};
\draw node[draw=none,align=left,anchor=north west] at (2.7,2.3) {$<\!p$};

\draw[-Latex] (1.45,1) to [out=-20,in=150] ++(0.8,0);
\draw[-Latex] (1.45,1.7) to [out=-25,in=195] ++(0.8,0);
\draw[-Latex] (1.55,2.3) to [out=30,in=200] ++(0.8,0);

\draw[-Latex] (1.5,0.5) to [out=10,in=190] ++(4.05,0);

\draw[-Latex] (4.75,1.0) to [out=30,in=200] ++(0.8,0);
\draw[-Latex] (4.75,1.6) to [out=-25,in=195] ++(0.8,0);
\draw[-Latex] (4.75,2.2) to [out=20,in=160] ++(0.8,0);

\end{tikzpicture}}\end{center}
\vspace{-1.5em} 

\caption{Schematic overview accompanying the proof of Lemma~\ref{lemma:str}. All cycles in $W'_\alpha$ consistent with $\sigma'_\alpha$ are in the $\textsf{Won}(\alpha)_{<p}$ regions, $\alpha=1-(p \bmod 2)$.}
\label{fig:regions}
\end{figure}

We partition $W'_\alpha$ into three regions $W_\alpha$, $\textsf{Won}(\alpha)_{<p}\setminus W_\alpha$, $(V_p \cap \textsf{Onestep}_\alpha)\setminus W_\alpha$.
We now prove that strategy $\sigma'_\alpha$ is such that all cycles consistent with $\sigma'_\alpha$ in $W'_\alpha$ are either fully inside $W_\alpha$ or fully inside $\textsf{Won}(\alpha)_{<p}\setminus W_\alpha$. No cycles are between the three regions or in the third region.
See also Fig.~\ref{fig:regions}.

\textit{First}, by induction hypothesis, we know that $\sigma_\alpha$ is such that $\alpha$ does not and $\invalpha$ cannot play from $W_\alpha$ to $\textsf{Won}(\invalpha)$ of the previous update, i.e., $V_{\leq p}\setminus W_\alpha$, which includes all vertices in $\textsf{Won}(\alpha)_{<p}\setminus W_\alpha$.
Hence, it is not possible to play from $W_\alpha$ to the other two regions.
\textit{Second}, by induction hypothesis, we know that $\sigma_{\alpha, <p}$ is such that no play is possible from $\textsf{Won}(\alpha)_{<p}$ to $\textsf{Won}(\invalpha)$, which includes the vertices in $V_p\setminus W_\alpha$.
This then also holds for the subregion $\textsf{Won}(\alpha)_{<p}\setminus W_\alpha$ with strategy $\sigma'_\alpha$.
\textit{Third}, any $\invalpha$-vertex in $V_p$ that is now a distraction only has successors already in $\textsf{Won}(\alpha)$ and $\alpha$-vertices in $V_p$ only choose successors already in $\textsf{Won}(\alpha)$.
Hence, no play stays in $V_p\cap \textsf{Onestep}_\alpha$ and therefore there are no cycles in $V_p\cap \textsf{Onestep}_\alpha$.
Since all cycles are only inside $W_\alpha$ or inside $\textsf{Won}(\alpha)_{<p}\setminus W_\alpha$,
we know by induction hypothesis about $\sigma'_\alpha$ that all cycles in $W_\alpha$ are won by $\alpha$ due to $\sigma_\alpha$ and that all cycles in $\textsf{Won}(\alpha)_{<p}\setminus W_\alpha$ are won by $\alpha$ due to $\sigma_{\alpha,<p}$.
Hence, all cycles consistent with $\sigma'_\alpha$ in $W'_\alpha$ are won by $\alpha$.

Furthermore, by induction hypothesis, we know that strategy $\sigma_\alpha$ does not allow a play from $W_\alpha$ to $\textsf{Won}(\invalpha)$; by induction hypothesis, we know that strategy $\sigma_{\alpha,<p}$ does not allow a play from $\textsf{Won}(\alpha)_{<p}\setminus W_\alpha$ to $\textsf{Won}(\invalpha)$; and by construction, every $p$-vertex that is now a distraction is either owned by $\invalpha$ and {cannot} play to $\textsf{Won}(\invalpha)$, or is owned by $\alpha$ with a strategy to play to $\textsf{Won}(\alpha)$ and does not play to $\textsf{Won}(\invalpha)$.

We have now proven that Lemma~\ref{lemma:str} holds after each iteration of the fixpoint; therefore it holds when the fixpoint computation is finished.

\textbf{Case 2. $\alpha = (p \bmod 2)$}.\qquad

After computing the fixpoint of $Z_p$,
we set $W_\alpha := \textsf{Won}(\alpha)_{<p} \cup (V_p\setminus Z_p)$ and choose as the strategy $\sigma_{\alpha,<p}$ and for $\alpha$-vertices in $V_p$ any successor in $\textsf{Won}(\alpha)$.
We know by induction hypothesis that $\sigma_{\alpha,<p}$ only allows cycles inside $\textsf{Won}(\alpha)_{<p}$ that are won by $\alpha$; furthermore all cycles inside $(V_p\setminus Z_p)$ are won by $\alpha$ with priority $p$ and the same holds for cycles between $\textsf{Won}(\alpha)_{<p}$ and $(V_p\setminus Z_p)$.
Furthermore, by induction hypothesis, there is no play from $\textsf{Won}(\alpha)_{<p}$ to $\textsf{Won}(\invalpha)$ and since the fixpoint of $Z_p$ is completed, we have a strategy in $\textsf{Won}(\alpha)$ for all $\alpha$-vertices in $(V_p\setminus Z_p)$ and no $\invalpha$-vertex in $(V_p\setminus Z_p)$ can play to $\textsf{Won}(\invalpha)$.
\end{proof}

\begin{lemma}
The \textup{\texttt{DFI}} algorithm solves parity games.
\end{lemma}

\begin{proof}
By Lemma~\ref{lemma:str},
after computing the fixpoint $Z_p$ of the highest priority $p$ of the game,
the regions $\textsf{Won}(\alpha)$ for players $\Even$ and $\Odd$ according to the set $Z$ 
are won by players $\Even$ and $\Odd$ respectively, as from each region the losing player cannot leave and the winner has a strategy to win all cycles inside the region.
\end{proof}


\subsection{Freezing below the fixpoint}

In the proof of Lemma~\ref{lemma:str},
we computed a strategy $\sigma_\alpha$ for player $\alpha$.
The obvious question is how we can modify \texttt{DFI}
to compute this strategy.
We propose an extension to \texttt{DFI} that 
follows easily from the proof of Lemma~\ref{lemma:str}.
After recomputing each lower fixpoint in the fixpoint iteration of a set $Z_p$,
notice that the vertices in $W_\alpha$, where $\alpha=1-(p \bmod 2)$, are still winning
for $\alpha$.
We do not need to recompute whether these vertices are distractions,
as we already know that they will be in the winning region of $\alpha$.
The only change will be that vertices currently winning for $\invalpha$ might become winning for $\alpha$.
We propose to use a set of \emph{frozen vertices} $F_p$ associated with each
fixpoint $Z_p$ that we do not recompute until fixpoint $Z_p$ has been
completed and we only recompute the lower fixpoints for vertices in
$V_{<p}\setminus F_p$.


\begin{lemma}
	\label{lemma:frozen}
	The winning region $W_\alpha$ ($\alpha\in\{\Even,\Odd\}$) after computing the fixpoint $Z_p$ is
	identical if we first fix a set of frozen vertices $F_\alpha\subseteq W_\alpha$ as won by player $\alpha$
	and then compute the fixpoint $Z_p$.
\end{lemma}
\begin{proof}
Proof by contradiction.	
Assume that some vertices $U\subseteq W_\alpha\setminus F_\alpha$ are won by $\invalpha$.
By Lemma~\ref{lemma:str}, we have that $\invalpha$ cannot reach $W_\invalpha$ from $U$, therefore $\invalpha$
has a strategy to win by staying inside $U$, \emph{regardless of the strategy of $\alpha$}.
However, by Lemma~\ref{lemma:str}, $\alpha$ has a strategy to win all plays that stay inside $W_\alpha$.
\end{proof}

As the set of frozen vertices $F_p$ is set to $W_\alpha$ before the next iteration of the lower fixpoint, and since $W_\alpha\subseteq W'_\alpha$, therefore $F_p\subseteq W'_\alpha$ and by Lemma~\ref{lemma:frozen} we know that recomputing the lower fixpoint while ``freezing'' vertices $F_p$ yields the same (correct) result.
Notice that the frozen vertices are exactly those of which we want to keep the strategy $\sigma_\alpha$ fixed.
Then whenever we evaluate which player wins a vertex in one step, we set the strategy accordingly, while all the frozen vertices keep their current strategy.
Since a vertex $v$ is in at most one frozen set at a time, we can also just employ a single function $F\colon V\to\{-,0,1,\dotsc,d\}$ to encode whether a vertex is frozen and at which priority.
This leads to Algorithm~\ref{alg:dfi}.

\begin{algorithm}[tb]
	\Def{\DFI{$\Game$}}{
		$Z \leftarrow V\mapsto 0$ \tcp*{start with no distractions}
		$F \leftarrow V\mapsto -$ \tcp*{start with no frozen vertices}
		$p \leftarrow 0$ \tcp*{start with lowest priority}
		\While(\tcp*[f]{while $\leq$ highest priority}){$p\leq d$}{
			$\alpha \leftarrow p \bmod 2$ \tcp*{current parity}
			$\textup{Chg} \leftarrow 0$ \;
			\ForAll(\tcp*[f]{evaluate non-frozen vertices}){$v\in V_p \colon F[v]=- \;\land\; Z[v]=0$}{
				$\alpha',\text{str}[v] \leftarrow \texttt{onestep}(v, Z)$ \tcp*{update strategy}
				\If(\tcp*[f]{new distraction?}){$\alpha'\neq \alpha$}{
					$Z[v] \leftarrow 1$ \tcp*{update fixpoint $Z_p$}
					$\textup{Chg} \leftarrow 1$ \tcp*{flag as changed}
				}
			}
			\If(\tcp*[f]{did $Z_p$ change?}){$\textup{Chg}$}{
				\ForAll(\tcp*[f]{freeze or reset lower vertices}){$v \in V_{<p} \colon F[v]=-$}{
					\lIf(\tcp*[f]{freeze vertex if won by $\invalpha$}){$\textsf{winner}(v,Z) = \invalpha$}{$F[v]\leftarrow p$}
					\lElse(\tcp*[f]{reset lower fixpoints}){$Z[v]\leftarrow 0$}
				}

				$p \leftarrow 0$ \tcp*{restart with lowest priority}
			}
			\Else{
				\lForAll(\tcp*[f]{thaw $v\in F_p$}){$v \in V_{<p} \colon F[v]=p$}{$F[v]\leftarrow -$}
				$p \leftarrow p + 1$ \tcp*{fixpoint, continue higher}
			}
		}
		$W_\Even, W_\Odd\leftarrow \{v \mid \textsf{winner}(v,Z)=0 \},\{v \mid \textsf{winner}(v,Z)=1 \}$ \;
$\sigma_\Even, \sigma_\Odd \leftarrow (v\in (W_\Even\cap V_\Even))\mapsto\textsf{str}[v],(v\in (W_\Odd\cap V_\Odd))\mapsto\textsf{str}[v]$ \;
\Return $W_\Even,W_\Odd,\sigma_\Even,\sigma_\Odd$
	}
	\Def{\Onestep{$v$, $Z$}}
	{
		$\alpha\leftarrow \textbf{if }v\in V_\Even\textbf{ then }0\textbf{ else }1$ \tcp*{obtain $\alpha$ the owner of $v$}
		\ForAll(\tcp*[f]{see if $\alpha$ can win in one step}){$u\in E(v)$}{
			\lIf(\tcp*[f]{if so, return winner and strategy}){$\textsf{winner}(u,Z)=\alpha$}{\Return $\alpha,u$}
		}
		\Return $\invalpha,-$ \tcp*{the opponent wins in one step, no strategy}
	}
\caption{The \texttt{DFI} algorithm extended to compute winning strategies}
\label{alg:dfi}
\vspace{-.2em}
\end{algorithm}

If we compare with Zielonka's recursive algorithm as presented in~\cite[Algorithm 3]{DBLP:conf/tacas/Dijk18},
we see that the frozen vertices in \texttt{DFI} are exactly the vertices in $W'_\invalpha$ that are not recomputed in the second recursion of line~10 of~\cite[Algorithm 3]{DBLP:conf/tacas/Dijk18}.
The recursive algorithm and \texttt{DFI} thus use the same mechanism to decide that a vertex is a distraction and to preserve the correct winning strategy.
The major difference between the two is that \texttt{DFI} relies on the one-step attractor while the recursive algorithm uses full attractor computation.
In fact, while the recursive algorithm requires repeated attractor computation and recursion,
\texttt{DFI} can be implemented with a simple loop,
as demonstrated in Algorithms~\ref{alg:dfi1}~and~\ref{alg:dfi}.

\subsection{Implementation}

We implement \texttt{DFI} in the parity game solver \textsc{Oink}~\cite{DBLP:conf/tacas/Dijk18}.
We use a bitvector to represent $Z$, recording 1 if a vertex is a distraction and 0 otherwise.
In addition,
we use a simple \texttt{int} array to represent $F$ as a function $V\to\{-,0,1,\dotsc,d\}$, as each vertex is only in at most one $F_p$ set at a time.
We also sort all vertices by priority after reading the input file,
so we can simply start with the first vertex and restart with the first vertex whenever we are done with a priority and the fixpoint is updated.
The implementation is available online via \url{https://www.github.com/trolando/oink}.


\begin{lemma}
	The \texttt{DFI} algorithm requires $\pazocal{O}(n\cdot\log(n))$ space and $\pazocal{O}(n^d)$ time.
\end{lemma}%
\begin{proof}
	For every vertex, we use 1 bit to mark whether the vertex is a distraction, furthermore $\lceil\log(d+1)\rceil$ bits for the priority at which a vertex is frozen and $\lceil\log(n)\rceil$ bits for the chosen strategy of each vertex.
	Hence the space complexity is $\pazocal{O}(n\cdot(1+\lceil\log(n)\rceil+\lceil\log(d+1)\rceil))=\pazocal{O}(n\cdot\log(n))$. 
	Whenever we reset lower fixpoints, a higher fixpoint monotonically increases. As there are $d+1$ fixpoints and each fixpoint can only be updated at most $n$ times, we obtain an upper bound of $\pazocal{O}(n^d)$.
\end{proof}

\subsection{Related work}

An optimization to only reset the fixpoints of the other parity between iterations has been proposed~\cite[Sec. 5.1]{DBLP:journals/corr/BruseFL14} by Bruse et al. and also applied~\cite{DBLP:journals/corr/abs-1809-03097} by Sanchez et al.
However our optimization is different and actually more powerful.
Not only do we not reset the lower fixpoints of the same parity, we actually do not re-evaluate all vertices with a lower priority that are currently won by player $1-(p \bmod 2)$ until the fixpoint $Z_p$ is computed,
which includes all vertices currently in the lower fixpoints $Z_{p-2}, Z_{p-4},\dotsc$ but also vertices of priorities $p-1,p-3,\dotsc$ that are not in $Z_{p-1},Z_{p-3},\dotsc$.

Hofmann et al.~\cite{DBLP:conf/spin/HofmannNR16} propose a way to compute winning strategies (called ``certificates'') directly for $\mu$-calculus model-checking which is very similar to our proposal here.
However they only compute the strategy for player $0$ and also have to maintain the witness strategies for all recursion depths of the fixpoint algorithm, resulting in a space complexity of $\pazocal{O}(\lvert S \rvert ^ 2 \lvert \phi \rvert^ 2)$, where $\lvert S \rvert$ is the size of the transition system and $\lvert \phi \rvert$ is the size of the $\mu$-calculus formula.
As the number of vertices of the parity game $n=\lvert S \rvert  \lvert \phi \rvert$, we can rewrite the space complexity of \texttt{DFI} as $\pazocal{O}(\lvert S \rvert  \lvert \phi \rvert \log(\lvert S \rvert  \lvert \phi \rvert))$ which improves upon~\cite{DBLP:conf/spin/HofmannNR16}.
This is a direct result of our optimization of freezing vertices between iterations of the fixpoint computation.


\section{Comparing with \texttt{APT} and \texttt{BFL}}
\label{sec:tworelated}

The \texttt{APT} algorithm originates with Vardi and Kupferman, who proposed an algorithm to solve parity games using weak alternating automata in~\cite{DBLP:conf/stoc/KupfermanV98}.
They extend parity games with two sets \textit{Visiting} and \textit{Avoiding}
meaning sets of states that are ``good for even'' and ``bad for even'' respectively. 
The \texttt{APT} algorithm was implemented and proven correct by Di Stasio et al~\cite{DBLP:conf/wia/StasioMPV16}.
They present \texttt{APT} for a \emph{minimum} parity condition, where the winner is determined by the lowest instead of the highest priority.
We initially present \texttt{APT} similarly, but rewrite to a \emph{maximum} parity condition at the end of a series of rewriting steps.
Furthermore, they denote with $V$ and $A$ the sets \textit{Visiting} and \textit{Avoiding}; we use $\Vis$ and $\Avoid$.
For consistency, we use $\mu$-calculus notation wherever appropriate.

The authors introduce the one-step attractor that computes all vertices in the game that can be attracted to a given set of vertices $X$ in one step,
\[
\textsf{force}_\alpha(X) := \{v\in V_\alpha \mid X \cap E(v) \neq \emptyset \} \cup \{v\in V_\invalpha \mid E(v) \subseteq X \}
\]

We can also write this one-step attractor using $\mu$-calculus notation,
\[
\textsf{force}_\alpha(X) := (V_\alpha \wedge \bdiamond X ) \lor (V_\invalpha \wedge \bbox X)
\]

Thus, $\textsf{force}_0$ attracts for player Even and $\textsf{force}_1$ for player Odd.
Also, if one player can force vertices to play to X, then we know that the opponent can force all other vertices to the rest of the game, i.e.,
\[
\textsf{force}_\alpha(X) = V\setminus\textsf{force}_\invalpha(V\setminus X)
\]

The \texttt{APT} algorithm is described as a fixpoint that is inductively defined\footnote{The published paper presents an incorrect definition, which has since been corrected.} to compute the set of winning states for Even $\textsf{Win}_0(V_1\cdot V_2 \cdots V_d, \emptyset, \emptyset)$ given a parity condition $\beta:=V_1\cdot V_2 \cdots V_d$ of the sets of states with priorities $1, 2, \dotsc, d$,
\[
\begin{tabu}{rcl}
\textsf{Win}_0(\varepsilon, \Vis, \Avoid) & := & \textsf{force}_0(\Vis) \\
\textsf{Win}_1(\varepsilon, \Avoid, \Vis) & := & \textsf{force}_1(\Avoid) \\
\textsf{Win}_0(V_i\cdot\beta', \Vis, \Avoid) & := & \mu Y^i . V \setminus \textsf{Win}_1(\beta', \Avoid\vee(V_i\setminus Y^i), \Vis\vee(V_i\land Y^i)) \\
\textsf{Win}_1(V_i\cdot\beta', \Avoid, \Vis) & := & \mu Y^i . V \setminus \textsf{Win}_0(\beta', \Vis\vee(V_i\setminus Y^i), \Avoid\vee(V_i\land Y^i)) \\
\end{tabu}
\]

%

The point of this definition is the alternation between computing $\textsf{Win}_0$ for sets $V_i$ of odd priorities and $\textsf{Win}_1$ for sets $V_i$ of even priorities.
The $Y^i$ set associated with each $V_i$ of odd priority thus describes the states that are good for Even and the $Y^i$ set for each $V_i$ of even priority describes all states that are good for Odd.
We now rewrite the definition step by step.
To clarify the next steps, we take as example the formulas for parity games with respectively 2, 3, 4 and 5 priorities:

\[
\begin{tabu}{rcl}
\textsf{Win}_0(V_1\cdots V_2) & := & \mu Y^1.V\setminus\mu Y^2.V\setminus\textsf{force}_0(V^3) \\
\textsf{Win}_0(V_1\cdots V_3) & := & \mu Y^1.V\setminus\mu Y^2.V\setminus\mu Y^3.V\setminus\textsf{force}_1(V^4) \\
\textsf{Win}_0(V_1\cdots V_4) & := & \mu Y^1.V\setminus\mu Y^2.V\setminus\mu Y^3.V\setminus\mu Y^4.V\setminus\textsf{force}_0(V^5) \\
\textsf{Win}_0(V_1\cdots V_5) & := & \mu Y^1.V\setminus\mu Y^2.V\setminus\mu Y^3.V\setminus\mu Y^4.V\setminus\mu Y^5.V\setminus\textsf{force}_1(V^6) \\
V^3 & := & (V_1\land Y^1)\lor (V_2\setminus Y^2) \\
V^4 & := & (V_1\setminus Y^1) \lor (V_2\land Y^2) \lor (V_3\setminus Y^3)\\
V^5 & := & (V_1\land Y^1) \lor (V_2\setminus Y^2) \lor (V_3\land Y^3) \lor (V_4\setminus Y^4) \\
V^6 & := & (V_1\setminus Y^1) \lor (V_2\land Y^2) \lor (V_3\setminus Y^3) \lor (V_4\land Y^4) \lor (V_5\setminus Y^5) \\
\end{tabu}
\]

See further~\cite{DBLP:conf/wia/StasioMPV16} for this derivation.
We use the earlier equivalence to eliminate $\textsf{force}_1$ and the fact that the sets $V_i$ partition $V$ to rewrite $V\setminus V^i$ for even $i$,
\[
\begin{tabu}{rcl}
\textsf{Win}_0(V_1\cdots V_2) & := & \mu Y^1.V\setminus\mu Y^2.V\setminus\textsf{force}_0(V^3) \\
\textsf{Win}_0(V_1\cdots V_3) & := & \mu Y^1.V\setminus\mu Y^2.V\setminus\mu Y^3.\textsf{force}_0(V\setminus V^4) \\
\textsf{Win}_0(V_1\cdots V_4) & := & \mu Y^1.V\setminus\mu Y^2.V\setminus\mu Y^3.V\setminus\mu Y^4.V\setminus\textsf{force}_0(V^5) \\
\textsf{Win}_0(V_1\cdots V_5) & := & \mu Y^1.V\setminus\mu Y^2.V\setminus\mu Y^3.V\setminus\mu Y^4.V\setminus\mu Y^5.\textsf{force}_0(V\setminus V^6) \\
V^3 & := & (V_1\land Y^1)\lor (V_2\setminus Y^2) \\
V\setminus V^4 & := & (V_1\land Y^1) \lor (V_2\setminus Y^2) \lor (V_3\land Y^3)\\
V^5 & := & (V_1\land Y^1) \lor (V_2\setminus Y^2) \lor (V_3\land Y^3) \lor (V_4\setminus Y^4) \\
V\setminus V^6 & := & (V_1\land Y^1) \lor (V_2\setminus Y^2) \lor (V_3\land Y^3) \lor (V_4\setminus Y^4) \lor (V_5\land Y^5) \\
\end{tabu}
\]

We can now rewrite these to:
\[
\begin{tabu}{rcl}
\textsf{Win}_0(V_1\cdots V_2) & := & \mu Y^1.V\setminus\mu Y^2.V\setminus\textsf{force}_0(\Vis) \\
\textsf{Win}_0(V_1\cdots V_3) & := & \mu Y^1.V\setminus\mu Y^2.V\setminus\mu Y^3.\textsf{force}_0(\Vis) \\
\textsf{Win}_0(V_1\cdots V_4) & := & \mu Y^1.V\setminus\mu Y^2.V\setminus\mu Y^3.V\setminus\mu Y^4.V\setminus\textsf{force}_0(\Vis) \\
\textsf{Win}_0(V_1\cdots V_5) & := & \mu Y^1.V\setminus\mu Y^2.V\setminus\mu Y^3.V\setminus\mu Y^4.V\setminus\mu Y^5.\textsf{force}_0(\Vis) \\
\Vis & := & \big( V_{\text{even}} \land \bigvee_{p=1}^{d}{(V_p\setminus Y^p)} \big) \cup \big( V_{\text{odd}} \land \bigvee_{p=1}^{d}{(V_p\wedge Y^p)} \big) \\
\end{tabu}
\]

Now, $\Vis$ is the set of states that is good-for-Even \emph{according to the fixpoint sets} $Y^1\dots Y^d$.
After computing these fixpoints, the odd sets contain all vertices won by player Even,
while the even sets contain all vertices won by player Odd.
Membership of $Y^p$ is only relevant in $\Vis$ for vertices with priority $p$.
When computing sets $Y^p$, we actually only need to record vertices with priority $p$.
The sets $Y^p$ thus encode that vertices in $V_p$ are good for the other player.
That is, that these vertices are \emph{distractions}.

\begin{lemma}
\texttt{APT} is equivalent to \texttt{DFI}.
\end{lemma}
\begin{proof}
Both algorithms compute the exact same sets $Y^p$, considering that only the vertices in $V_p$ are relevant for each $Y^p$ and we therefore only update $Y^p$ for vertices in $V_p$.
The algorithms both compute first the fixpoint of distracting vertices of the least important priority, then continue with the next priority.
Each time the fixpoint of some priority is updated with new distracting vertices, all deeper nested fixpoints are reset.
Thus, both algorithms compute the same sets of distracting vertices.
\end{proof}

Based on the formulas by Walukiewicz~\cite{DBLP:conf/stacs/Walukiewicz96} that translate the winning condition of a parity game to $\mu$-calculus formulas over the parity game, Bruse et al. proposed a fixpoint algorithm we call \texttt{BFL}~\cite{DBLP:journals/corr/BruseFL14}.

\begin{lemma}
\texttt{APT} is equivalent to \texttt{BFL}.
\end{lemma}
\begin{proof}
We define $X_i$ as $Y^i$ for odd $i$ and $V\setminus Y^i$ for even $i$, i.e., $X_i$ represents the states that are good-for-Even. Rewriting $Y^i:=V\setminus X_i$ and $\mu Y^i.\phi:=V\setminus\nu X_i.V\setminus \phi$ for even $i$, and simply $Y^i:=X_i$ for odd $i$, we obtain
\[
\begin{tabu}{rcl}
\textsf{Win}_0(V_1\cdots V_2) & := & \mu X_1.\nu X_2.\textsf{force}_0(\Vis) \\
\textsf{Win}_0(V_1\cdots V_3) & := & \mu X_1.\nu X_2.\mu X_3.\textsf{force}_0(\Vis) \\
\textsf{Win}_0(V_1\cdots V_4) & := & \mu X_1.\nu X_2.\mu X_3.\nu X_4.\textsf{force}_0(\Vis) \\
\textsf{Win}_0(V_1\cdots V_5) & := & \mu X_1.\nu X_2.\mu X_3.\nu X_4.\mu X_5.\textsf{force}_0(\Vis) \\
\Vis & := & \bigvee_{p=1}^{d}{(V_p\wedge X_p)} \\
\end{tabu}
\]

Furthermore, because the sets $V_p$ partition $V$, we have that
\[
\bigvee_{p=1}^{d}{(V_p\wedge X_p)} \equiv \bigwedge_{p=1}^{d}{(V_p\to X_p)} \equiv \bigwedge_{p=1}^{d}{(\neg V_p \lor X_p)}
\]

We use the definition of $\textsf{force}_0(X)$ to obtain
\[
\begin{tabu}{rcl}
\textsf{Win}_0 & := & \mu X_1.\nu X_2\dots\sigma X_d.\big( (V_\Even \wedge \bdiamond \bigvee_{p=1}^{d}{(V_p\wedge X_p)} ) \lor (V_\Odd \wedge \bbox \bigvee_{p=1}^{d}{(V_p\wedge X_p)}) \big) \\
& = & \mu X_1.\nu X_2\dots\sigma X_d.\big( (V_\Even \wedge \bdiamond \bigvee_{p=1}^{d}{(V_p\wedge X_p)} ) \lor (V_\Odd \wedge \bbox \bigwedge_{p=1}^{d}{(\neg V_p\lor X_p)}) \big) \\
& = & \mu X_1.\nu X_2\dots\sigma X_d.\big( (V_\Even \wedge \bigvee_{p=1}^{d}{\bdiamond(V_p\wedge X_p)} ) \lor (V_\Odd \wedge \bigwedge_{p=1}^{d}{\bbox(\neg V_p\lor X_p)}) \big)
\end{tabu}
\]
where $\sigma$ is $\nu$ if $d$ is odd and $\mu$ if it is even.

Finally, we change from the \emph{minimal} parity condition $V_1\cdots V_d$ to the \emph{maximal} parity condition $V_{d-1}\cdots V_0$ and then our final result is
$$\textsf{Win}_0 := \sigma X_{d-1}\dots\mu X_1.\nu X_0.\big( (V_\Even \wedge \bigvee_{p=0}^{d-1}{\bdiamond(V_p\wedge X_p)} ) \lor (V_\Odd \wedge \bigwedge_{p=0}^{d-1}{\bbox(\neg V_p\lor X_p)}) \big)$$

This is precisely the fixpoint formula in~\cite{DBLP:journals/corr/BruseFL14}.
\end{proof}

In a footnote, Di Stasio et al write: ``The unravelings of $\textsf{Win}_0$ and $\textsf{Win}_1$ have some analogies with the fixed-point formula introduced in~\cite{DBLP:journals/corr/BruseFL14} also used to solve parity games.
Unlike our work, however, the formula presented there is just a translation of the Zielonka algorithm~\cite{DBLP:conf/stacs/Walukiewicz96}.''
We have now shown that these algorithms are equivalent.
Furthermore, as the computed fixpoints of the three algorithms are equivalent, so are their implementations. The only major difference is that \texttt{APT} does not find the winning strategies and \texttt{BFL} has a rather convoluted way to find the winning strategies.

\section{Empirical evaluation}
\label{sec:evaluation}

\begin{table}[bt]
	\begin{tabu} to \linewidth {X[1.4l]|X[r]X[1.2r]|X[r]X[r]|X[r]X[r]}
		\toprule
		& \multicolumn{2}{c|}{\textbf{equivalence checking}} & \multicolumn{2}{c|}{\textbf{model-checking}} & \multicolumn{2}{c}{\textbf{reactive synthesis}}\\
		\midrule
		priorities & \multicolumn{2}{c|}{2} & \multicolumn{2}{c|}{1--4} & \multicolumn{2}{c}{3--12} \\
		count & \multicolumn{2}{c|}{216} & \multicolumn{2}{c|}{313} & \multicolumn{2}{c}{223} \\
		\midrule
		& \textbf{mean} & \textbf{max} & \textbf{mean} & \textbf{max} & \textbf{mean} & \textbf{max} \\
		\# vertices & 3,288,890 & 40,556,396 & 866,289  & 27,876,961 & 921,484 & 31,457,288 \\
		\# edges & 10,121,422 & 167,527,601 & 2,904,500 & 80,830,465 & 1,693,544 & 59,978,691 \\
		avg. outdegree & 2.35 & 5.09 & 2.75 & 6.14 & 1.68 & 2.00 \\
		\bottomrule
	\end{tabu}	
	\caption{Statistics of the three benchmark sets used for the empirical evaluation.}
	\label{tbl:benchmarks}
\end{table}

The goal of the empirical evaluation is to study the performance of the version of \texttt{DFI} that computes winning strategies and to compare it with its closest cousins, the attractor-based algorithms Zielonka (\texttt{ZLK}), priority promotion (\texttt{PP}) and tangle learning (\texttt{TL}). We also compare with strategy iteration (\texttt{SI}) since this algorithm is used as a backend by the LTL synthesis tool \textsc{Strix}~\cite{DBLP:conf/cav/MeyerSL18}.

We do not report on the various crafted benchmarks designed to expose pathological behavior,
as it is clear that \texttt{DFI} is very much vulnerable to such artificial games.
We also do not report on random games with few priorities, as this has been done before with positive results~\cite{DBLP:conf/wia/StasioMPV16}, however this comparison has two problems.
\textit{First}, the comparison was done in \textsc{PGSolver}~\cite{DBLP:conf/atva/FriedmannL09},
which is certainly an important platform for parity game research,
but is also significantly slower than the implementations in \textsc{Oink}~\cite{DBLP:conf/tacas/Dijk18}.
\textit{Second}, there is no obvious reason why \textit{random} games with few priorities would be a good representative for games derived from actual applications.
Hence, we implement \texttt{DFI} in the parity game solver \textsc{Oink}
and we use as benchmarks those from model-checking and equivalence
checking proposed by Keiren~\cite{DBLP:conf/fsen/Keiren15}.
These are 313 model-checking and 216 equivalence checking games. 
Furthermore, we consider a new category of 223 ``reactive synthesis'' benchmarks obtained via the LTL synthesis tool \textsc{Strix}~\cite{DBLP:conf/cav/MeyerSL18} from the synthesis competition~\cite{DBLP:journals/corr/abs-1904-07736}.
See also Table~\ref{tbl:benchmarks}.
Notable is the very high number of vertices compared to the number of priorities,
and also that the average outdegree of vertices in these games is very low.


The experiments were performed on a cluster of Dell PowerEdge M610 servers with two Xeon E5520 processors and 24~GB internal memory each.
The tools were compiled with gcc 5.4.0. 
All experimental scripts and log files are available via \url{https://www.
github.com/trolando/dfi-experiments}.


\begin{table}[tbp]
\begin{tabu}{X[3.0l]|X[r]X[r]X[r]X[r]X[r]@{$\;\;$}|X[r]X[r]X[r]X[r]X[1.2r]}
	\toprule
	\textbf{Dataset} & \multicolumn{5}{c|}{\textbf{with preprocessing}} & \multicolumn{5}{c}{\textbf{without preprocessing}} \\	
	& \texttt{dfi} & \texttt{zlk} & \texttt{pp} & \texttt{tl} & \texttt{si} & \texttt{dfi} & \texttt{zlk} & \texttt{pp} & \texttt{tl} & \texttt{si} \\
	\midrule
	equivalence & 401 & \textcolor{blue}{\textbf{381}} & 389 & 455 & 6218 & 970 & 470 & 444 & 570 & 19568 \\
	model-checking & \textcolor{blue}{\textbf{59}} & 73 & 82 & 166 & 292 & 156 & 79 & 93 & 182 & 2045 \\
	synthesis & 62 &\textcolor{blue}{\textbf{52}} & 57 & 59 & 158 & 64 & \textcolor{blue}{\textbf{51}} & 70 & 67 & 175 \\
	\bottomrule
\end{tabu}

\caption{Cumulative time in sec. (average of five runs) spent to solve all games in each set of benchmarks, with a timeout of 1800 seconds. We record 1800 seconds when the computation timed out.}
\label{tbl:results}
\vspace{-.5em}
\end{table}

Table~\ref{tbl:results} shows the cumulative runtimes of the five algorithms, with and without the preprocessing analysis (removing self-loops and winner-controlled winning cycles).
We record the total time spent solving, including preprocessing.
Only \emph{strategy iteration} had timeouts. 
For the runs that timed out, we simply used the timeout value of $1800$ seconds, but this underestimates the actual runtime.
See further~\cite{DBLP:conf/tacas/Dijk18} for a comparison that shows that \texttt{zlk} and \texttt{pp} are the fastest solvers for practical games.
Notice that preprocessing is helpful in almost all cases.
A remarkable result is that the \texttt{DFI} algorithm (with preprocessing) is the fastest solver for model-checking games.
We learn from closer inspection of the data that almost all games are solved in a fraction of a second by all solvers,
while \texttt{DFI} slightly outperforms the other solvers for the handful of slower games.
An explanation for the speed of \texttt{DFI} probably lies in its simplicity and in favorable memory access patterns.
As the vertices are ordered by priority before solving, vertices are evaluated consecutively for the fixpoints,
which could be more efficient than when for example use of full attractors results in irregular memory accesses.
Furthermore, the implementation in C is a tight loop without recursion.
We also parallelized all three \textbf{forall} loops in Algorithm~\ref{alg:dfi} and obtain good speedups of up to 4.8$\times$ with 8 cores; however most games are already solved within seconds. In the interest of space, we omit these results.

\section{Conclusions}
\label{sec:conclusions}

We have discussed distractions in parity games and how various algorithms deal with the distractions.
It is important to the understanding of parity game solvers to know how they deal with distractions, that is, how they decide that a vertex is a distraction or no longer a distraction and why.
We show a fundamental difference between attractor-based algorithms and algorithms employing progress measures.

We have implemented a new fixpoint algorithm \texttt{DFI} that computes winning strategies by ``freezing'' the winning strategy that is the witness to the decision that certain vertices are now a distraction.

The remarkable result that \texttt{DFI} is the fastest solver for parity games from model-checking leads us to wonder whether parity games might be a distraction for solving $\mu$-calculus model-checking, as simple fixpoint iteration already selects successful strategies.

\clearpage

\bibliographystyle{eptcs}
\bibliography{lit}

\end{document}